\documentclass[lettersize,journal]{IEEEtran}
\usepackage{amsmath,amsfonts}
\usepackage{array}
\usepackage[font=normalsize,labelfont=sf,textfont=sf]{subfig}
\usepackage{textcomp}
\usepackage{stfloats}
\usepackage{verbatim}
\usepackage{cite}

\usepackage{hyperref}
\usepackage[utf8]{inputenc} % allow utf-8 input
\usepackage[T1]{fontenc}    % use 8-bit T1 fonts
\usepackage{hyperref}       % hyperlinks
\usepackage{url}            % simple URL typesetting
\usepackage{booktabs}       % professional-quality tables
\usepackage{amsfonts}       % blackboard math symbols
\usepackage{nicefrac}       % compact symbols for 1/2, etc.
\usepackage{microtype}      % microtypography
\usepackage{color}          % color
\usepackage{xcolor}
\usepackage{graphicx}
\usepackage{mathrsfs}
\usepackage{autobreak}
\usepackage{bm}
\usepackage{bbm}
\usepackage{multirow} 
\usepackage{multicol} 
\usepackage{arydshln}
\usepackage{graphicx}
\usepackage{latexsym}
\usepackage{caption} 
\usepackage{booktabs}
\usepackage{bbding}

\usepackage[ruled,linesnumbered]{algorithm2e}
\usepackage[english]{babel}
\usepackage[shortlabels]{enumitem}

\begin{document}

\newcommand{\eg}{e.g.}
\newcommand{\ie}{i.e.}
\newcommand{\etc}{etc}
\newcommand{\wrt}{w.r.t.~}
\newtheorem{prop}{Proposition}
\newtheorem{proof}{Proof}

\title{Partial Relaxed Optimal Transport for \mbox{Denoised Recommendation}}

% \author{Corey Gray\thanks{Society for Industrial and Applied Mathematics.}
% \and Tricia Manning\thanks{Society for Industrial and Applied Mathematics.}}
\author{
        Yanchao~Tan,
        Carl~Yang~\IEEEmembership{Member,~IEEE}, 
        Xiangyu~Wei,
        Ziyue~Wu,
        Xiaolin~Zheng~\IEEEmembership{Senior Member,~IEEE}
        \IEEEcompsocitemizethanks{
        \IEEEcompsocthanksitem Y. Tan,  X. Wei, and X. Zheng (Corresponding Author) are with 
        the College of Computer Science, Zhejiang University, Hangzhou 310027, China, E-mail:\{yctan, weixy, xlzheng\}@zju.edu.cn.
        \IEEEcompsocthanksitem C. Yang is with the Department of Computer Science, Emory University, Atlanta 30322, United State, E-mail: j.carlyang@emory.edu.
        \IEEEcompsocthanksitem Z. Wu is with the School of Management, Zhejiang University, Hangzhou 310058, China, E-mail: wuziyue@zju.edu.cn.
        
        }
}

\markboth{Journal of \LaTeX\ Class Files,~Vol.~14, No.~8, August~2021}%
{Shell \MakeLowercase{\textit{et al.}}: A Sample Article Using IEEEtran.cls for IEEE Journals}


\maketitle

% Copyright Statement
% When submitting your final paper to a SIAM proceedings, it is requested that you include
% the appropriate copyright in the footer of the paper.  The copyright added should be
% consistent with the copyright selected on the copyright form submitted with the paper.
% Please note that "20XX" should be changed to the year of the meeting.

% Depending on which copyright you agree to when you sign the copyright form, the copyright
% can be changed to one of the following after commenting out the default copyright statement
% above.

%\fancyfoot[R]{\scriptsize{Copyright \textcopyright\ 20XX\\
%Copyright for this paper is retained by authors}}

%\fancyfoot[R]{\scriptsize{Copyright \textcopyright\ 20XX\\
%Copyright retained by principal author's organization}}

%\pagenumbering{arabic}
%\setcounter{page}{1}%Leave this line commented out.

\begin{abstract} \small\baselineskip=9pt 
The interaction data used by recommender systems (RSs) inevitably include noises resulting from mistaken or exploratory clicks, especially under implicit feedbacks. Without proper denoising, RS models cannot effectively capture users' intrinsic preferences and the true interactions between users and items. To address such noises, existing methods mostly rely on auxiliary data which are not always available. In this work, we ground on Optimal Transport (OT) to globally match a user embedding space and an item embedding space, allowing both non-deep and deep RS models to discriminate intrinsic and noisy interactions without supervision. Specifically, we firstly leverage the OT framework via Sinkhorn distance to compute the continuous many-to-many user-item matching scores. Then, we relax the regularization in Sinkhorn distance to achieve a closed-form solution with a reduced time complexity. Finally, to consider individual user behaviors for denoising, we develop a partial OT framework to adaptively relabel user-item interactions through a personalized thresholding mechanism. Extensive experiments show that our framework can significantly boost the performances of existing RS models.
\end{abstract}

\begin{IEEEkeywords}
Recommender systems, denoised recommendation, optimal transport, implicit feedback
\end{IEEEkeywords}

\section{Introduction}
\label{sec:intro}
% RS denoising
% To achieve this goal, most RSs only show users a narrow subset of the entire range of available recommendations. However, the subset inevitably contains some noise items that not relate to user's intrinsic preference but the ones RSs think users like.
% Once users accidentally interact with the noise items, the noise items will be added to the subset of the most relevant items. 
% Therefore, based on the mechanism that shows the most relevant items in the subset, such ``relative” items become even easier to be recommended. 
With the rapid growth of various activities on the Web, recommender systems (RSs) become fundamental in helping users alleviate the problem of information overload. 
% To recommend proper items to users, RSs try to project users and items into two spaces and measure user-item similarities across the two spaces.
However, users can click some items by mistake or out of curiosity, and many RSs will also recommend some less relevant items for exploration every now and then. 
%As a result, the user-item interaction data used by the RS models inevitably include noises, which do not exactly reflect the users' intrinsic preferences.
%Compared with explicit feedbacks (\eg, 1-5 star ratings), implicit feedbacks (\eg, clicks and browsing history) are even more prone to the noises caused by mistaken or exploratory clicks.
Take movie watching for example. Since a user cannot always distinguish \textsf{horror} movies from \textsf{romantic} ones by movie names, he/she can easily click a \textsf{horror} movie by mistake among the many \textsf{romantic} movies he/she usually watches. In this case, a traditional method like CF cannot effectively reduce the probability of recommending a \textsf{horror} movie unless the noisy clicks are extremely rare, which may further seduce even more noisy clicks.
%Without a proper denoising mechanism, the noisy feedbacks can even lead the RSs to make increasingly wrong recommendations over time.

% 没人考虑这样的noise, MF DMF都会受这个影响 都会落入这个suboptimal
% 所以需要denoising
% However, many traditional RS strategies (\eg, Matrix Factorization (MF)) ignore the existence of noises and try to equally fit all user-item interactions. 
% As shown in Figure \ref{fig:illustration}, compared with user B whose preferences are evenly distributed among \textsf{romance} and \textsf{horror} categories, most of user A's interactions are \textsf{romantic} movies. 
% User A might also accidentally click on one \textsf{horrible} movie by mistake or out of curiosity. 
% In this case, a RS would model user A's preference as similar to user B's, and thus keep recommending \textsf{horrible} movies to her/him afterwards, due to the failure in considering the discrimination between intrinsic preferences and noises. 

Recently, to get out of such mistakes, existing research discovers users' intrinsic preferences under the help of external user behaviors \cite{kim2014modeling}, auxiliary item features \cite{lu2019effects} or extra feedbacks \cite{yang2012exploiting}. A key limitation with them is their reliance on auxiliary data, which costs significant effort and is not always attainable in RSs.
Without supervision, can we develop a framework to automatically denoise user-item interactions for a more accurate recommendation?

In this work, we approach the denoised recommendation problem in two steps: 
\begin{enumerate}
\item distinguishing intrinsic and noisy interactions; 
\item stressing intrinsic preferences and reducing noisy ones for recommendation. 
\end{enumerate}
Specifically, inspired by a least modeling effort principle in an unsupervised fashion \cite{DBLP:journals/pami/CourtyFTR17,li2019learning,wang2020denoising}, we refer to noises as \textit{minor abnormal data that needs higher modeling effort} than the intrinsic ones. 

Based on the above rationale, we propose to distinguish noises by ranking the global matching cost between user and item embedding spaces, and flexibly integrate it with both non-deep and deep RS methods.
In this way, the key to distinguishing noises boils down to finding a global matching matrix with the minimum matching cost.
Subsequently, we advocate a principled denoised RS framework and calculate the matching matrix grounding on \textit{Optimal Transport} (OT), which has been introduced to address the unsupervised matching problem between two distributions or spaces in many fields \cite{bhushan2018deepjdot,huynh2020otlda}. 
Through minimizing the overall cost of the mismatched user-item pairs, OT finds a global optimal matching matrix, whose values represent the matching cost (\ie, modeling effort). Specifically, \textit{the low user-item matching cost indicates a user's intrinsic preferences while the high value indicates noisy ones}.

% Common practices of OT usually optimize for a one-hot matching matrix, due to the requirements of the most widely studied tasks in transfer learning \cite{bhushan2018deepjdot,xujoint}.
However, to apply OT to achieve denoised recommendation, three problems still remain: 
\begin{enumerate}
    \item Unlike the one-hot constraint in previous applications of OT, the nature of RSs requires a many-to-many matching; 
    \item The large number of users and items makes the matching process time-consuming; 
    \item A mechanism is needed to properly leverage the learned matching matrix and finally eliminate the effect of noises for denoised recommendation.
\end{enumerate}
%Moreover, to differentiate intrinsic and noisy interactions for each user through ranking, the interaction measurement needs to be continuous.
%Therefore, ordinary discrete OT cannot be directly applied to RSs. 

To address these challenges, we first leverage the OT framework via Sinkhorn distance  \cite{cuturi2013sinkhorn}. Unlike a one-hot matching matrix that satisfies the requirements of the many widely studied tasks (\eg, transfer learning \cite{bhushan2018deepjdot,xujoint}), we compute a continuous many-to-many approximation to the discrete OT, so as to meet the nature of RSs.
Furthermore, we propose to relax the regularization in Sinkhorn to achieve a closed-form solution, which can reduce the time complexity for large-scale data from $\mathcal{O}(\max(M, N)^3)$ to $\mathcal{O}(\max(M, N)^2)$ .
Then, we rank the learned continuous matching cost to differentiate intrinsic and noisy interactions for each user.
Finally, we take into account that individual users' behaviors can vary in RSs, which should be handled differently during denoising. To this end, we design a personalized thresholding mechanism for relabeling user-item interactions, which efficiently finds a splitting point for each individual user according to the ranked matching cost. The whole Partial relaxed optimal Transport for denoised Recommendation (ProRec).  
% \red{Furthermore, ProRec can be integrated for both deep and non-deep methods for denosing.}
% For example, in Figure \ref{fig:illustration}, the threshold for user B with diverse behaviors is 5 and that for user C with consistent behaviors is 3. 
%Such matching cost-based ranking and thresholding constitute our \blue{ProRec} framework, which can be trained iteratively with existing RS models to adaptively relabel user-item interactions for denoised recommendation. 

Experiment results on one real-world dataset with synthesized noises demonstrate our proposed ProRec's ability of denoising under various levels of noises. Further experiments on three original real-world datasets also verify the advantages of ProRec, which demonstrate its significant performance boosts upon multiple popular RS models, in terms of effectiveness and efficiency.

Our main contributions are summarized as follows:
\begin{itemize}
    \item We propose a novel denoised recommendation framework based on the Optimal Transport (OT) theory, which discriminates intrinsic and noisy user-item interactions without supervision. 
    \item We propose a ProRec framework to meet the nature of RSs and reduce the time complexity for large-scale data from $\mathcal{O}(\max(M, N)^3)$ to $\mathcal{O}(\max(M, N)^2)$, by ranking and thresholding the many-to-many matching matrix between users and items.
    \item We conduct extensive experiments on four public datasets to verify the efficacy of our proposed framework. The experimental results demonstrate strong performances against state-of-the-art RS models. 
\end{itemize}

\section{Related work}
In this section, we provide literature review on two relevant lines, \ie, recommender systems and optimal transport.
\subsection{Recommender Systems}
% missingness in implicit feedback make the exposure bias worse
% In recent years, implicit feedback (e.g., click and watch) has been widely used as indications of user preference in recommendation systems \cite{he2017neural} since the collection of explicit feedback is time-consuming and possibly expensive \cite{hu2008collaborative,wang2020denoising}. 
% Even though implicit feedback is somehow correlated with user preference, 
Many research studies have pointed out that there exist a large proportion of noisy interactions in the implicit feedbacks (\ie, the clicks) \cite{hu2008collaborative,kim2014modeling}. These feedbacks are easily affected by different factors, such as the position bias \cite{jagerman2019model}, caption bias \cite{lu2018between} and exposure bias \cite{hu2008collaborative,schnabel2016recommendations}. Therefore, there exist large gaps between the implicit feedbacks and the actual user preferences due to various reasons, which are detrimental to the users' continuous behaviors and overall satisfaction \cite{lu2019effects,wang2020denoising}.
% \cite{wang2020denoising} found serious negative impacts of noisy implicit feedback, i.e., fitting the noisy data prevents the recommender from learning the actual user preference

% auxilary data, MF propensity
To alleviate the above issues, many researchers have considered incorporating auxiliary feedbacks to the identification of noises in the implicit feedbacks, such as dwell time \cite{kim2014modeling} and scroll intervals \cite{lu2018between}. 
These works usually design features manually and label items with external help (\eg, from domain experts) \cite{lu2018between,lu2019effects}, which require extensive effort and cost, which are not always available across RSs.
To this end, sheer observations on training losses have recently been explored to denoise the implicit feedbacks during training in an unsupervised fashion \cite{wang2020denoising}. The method is purely heuristic and can be integrated with deep learning methods only.
Based on a principle of the least modeling effort in different fields \cite{DBLP:journals/pami/CourtyFTR17,jiang2018mentornet}, our proposed framework aims to denoise the user-item interactions without supervision, which is supported by Optimal transport (OT) theory with guaranteed optimality in the global matching and can be flexibly integrated with both deep and non-deep RS methods.

\subsection{Optimal Transport} 
OT aims to find a matching
between two distributions or spaces, which has been most widely used in transfer learning. It can match one instance in the source domain to another in the target domain with the least cost. Existing works mainly study discrete OT for one-to-one matching. For example, in computer vision, OT has been applied to domain adaption \cite{DBLP:journals/pami/CourtyFTR17,bhushan2018deepjdot} and style transfer \cite{kolkin2019style}. OT has also been successfully applied in many natural language processing tasks, such as topic modeling \cite{huynh2020otlda,kusner2015word}, text generation \cite{chen2018adversarial}, and sequence-to-sequence learning \cite{DBLP:conf/iclr/ChenZZTGZLSCC19}. 
Considering the effect of non-matched pairs, \cite{xujoint} proposed a partial one-to-one matching that leverages one shared threshold to address domain shift for domain adaption instead of taking all pairs into account. 

However, the above discrete one-to-one matching cannot satisfy the nature of many-to-many user-item relationships in RSs, where each user can interact with many items and each item can interact with many users.
Though OT is applied into cold-start recommendation problem in \cite{meng2020wasserstein}, they directly measure user-item similarity without taking the effect of noises into consideration. 
Moreover, partial matching with a shared threshold cannot reflect the individual user behaviors, and thus leads to suboptimal recommendation performances.

% \newpage
\begin{figure*}
    \centering
    \includegraphics[width=0.95\linewidth]{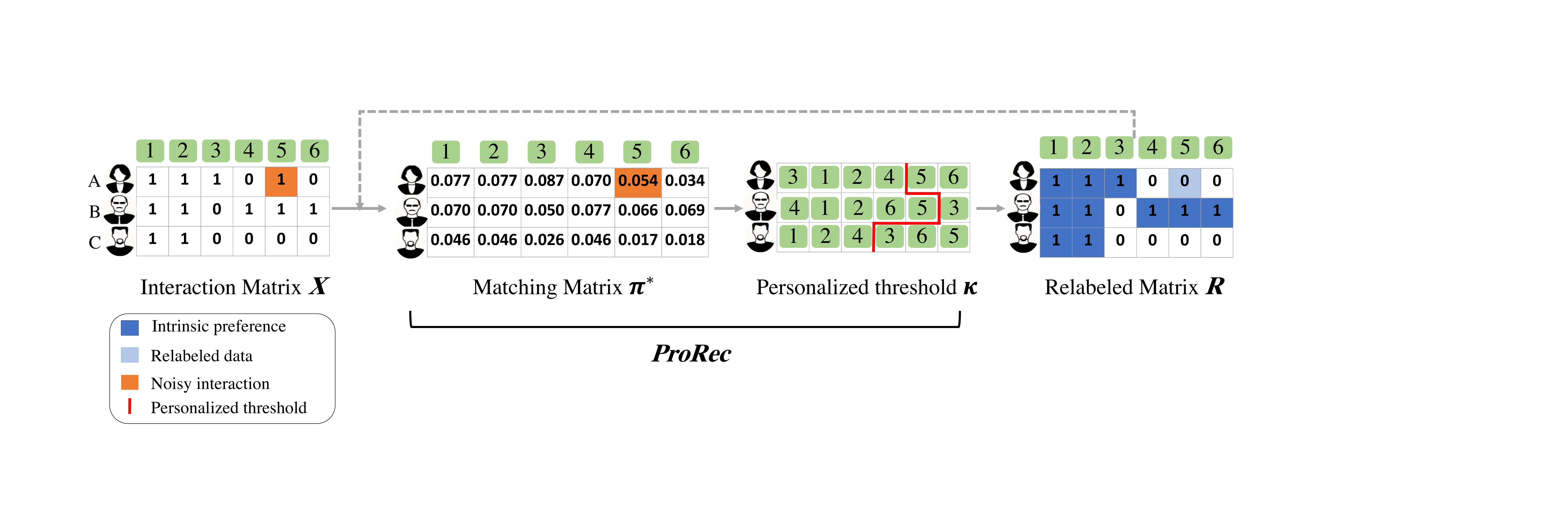}
    \caption{A toy example illustrating the denoised recommendation based on our ProRec framework.}
    \label{fig:model}
\end{figure*}
\section{Methodology}
In this section, we detail the proposed Partial relaxed optimal Transport for denoised Recommendation (ProRec) framework, as shown in Figure \ref{fig:model}. 
Specifically, we first formulate the denoised recommendation as a two-step process.
Then, we introduce the Optimal Transport (OT) framework for distinguishing noises through a global user-item matching. 
Furthermore, we elaborate our proposed ProRec framework, which can be applied to achieve denoised recommendation.
% extend the Sinkhorn algorithm with relaxed regularization towards many-to-many matchings for continuous ranking and efficient OT computation.
% Finally, we elaborate a novel Partial OT framework to adaptively relabel user-item interactions for individual behaviors modeling with personalized thresholding.

\subsection{Problem Statement}
Following the setting of recommendation, we denote the user-item interaction matrix as $\bm{X} \in \{0, 1 \}^{M \times N}$, where $M$ denotes the number of users and $N$ denotes the number of items. We construct user-item relationships across a \textit{user embedding space} $\mathcal{D}^u$ and an \textit{item embedding space} $\mathcal{D}^v$. The problem is that user-item pairs contain some noisy interactions which need to be relabeled. 

In this work, we approach denoised recommendation through two steps. Firstly, we distinguish intrinsic and noisy interactions across $\mathcal{D}^u$ and $\mathcal{D}^v$ by a global matching matrix $\bm{\pi}^*$.
Then, to stress users' intrinsic preferences and reduce noises, we rank the learned $\bm{\pi}^*$ and keep only top $\bm{\kappa}$ interactions of each user via a personalized thresholding mechanism. In this way, we learn a relabeled matrix $\bm{R} \in R^{M \times N}$ for denoised recommendation.
% To obtain the relabeled matrix $\bm{R} \in R^{M \times N}$, the critical task is to figure out the matching matrix $\bm{\pi}^*$ which can indicate user's intrinsic preferences and the true relationships between $\mathcal{D}^u$ and $\mathcal{D}^v$.

\subsection{Global Matching for Distinguishing Noises}
As motivated in Section \ref{sec:intro}, a framework that can denoise user-item interactions for more accurate recommendation is in great need. 
However, without supervision, it is hard to define noisy interactions and then stress intrinsic preferences together while reducing noisy ones.

To distinguish noises in an unsupervised fashion, several works \cite{DBLP:journals/pami/CourtyFTR17,jiang2018mentornet,xu2018social,wang2020denoising} have pointed out a principled way to locate noises by finding the data that needs the most modeling effort or cost.
Inspired by this principle of least modeling effort, in this work, we propose to refer to the noises as minor abnormal data. Compared with the interactions from intrinsic preferences, noises take much effort or cost to model.

Based on the above least effort rationale, we propose a novel denoising framework to distinguish noises by ranking the global matching between user and item embedding spaces. 
The framework is served as a plug-in, so as to be flexibly integrated to both non-deep and deep RS methods.
In this scenario, the key to distinguishing noises boils down to finding a global user-item matching across the user space $\mathcal{D}^u$ and the item space $\mathcal{D}^v$ with minimal matching cost. 

To address the unsupervised matching problem between two spaces, we advocate a principled denoised RS framework and ground on \textit{Optimal Transport} (OT), which has been successfully applied in many fields such as CV \cite{bhushan2018deepjdot} and NLP \cite{huynh2020otlda}.
As studied theoretically in \cite{gretton2012optimal}, 
the OT formulation derived from the Earth Movers Distance (EMD) can be expressed as
\begin{equation}
\begin{aligned}
    &\min_{\bm{\pi} \in \mathcal{X}(\mathcal{D}^u,\mathcal{D}^v)} \sum_{i=1}^M \sum_{j=1}^N \mathcal{\bm{M}}_{ij}\pi_{ij}, \\
    s.t. \,\, &\bm{1}_M\bm{\pi} = \frac{1}{N} \bm{1}_N, \bm{1}_N\bm{\pi}^T = \frac{1}{M}\bm{1}_M, \\
\end{aligned}
\label{eq:emd}
\end{equation}
where $\mathcal{X}(\mathcal{D}^u,\mathcal{D}^v)$ denotes the joint probability distribution of user embedding space $\mathcal{D}^u$ and item embedding space $\mathcal{D}^v$. $\bm{1}_M \in \mathbb{R}^{1\times M},\bm{1}_N \in \mathbb{R}^{1\times N}$ with all elements equal to 1. $M$ is the number of users and $N$ is the number of items. $\bm{\pi}^*$ denote the optimal transport plan.

Specifically, the matching cost matrix $\mathcal{M}$ directly reflects the modeling effort, and thus the low user-item matching cost can indicate users' intrinsic preferences while the high values can indicate noises. 
Based on the definition of $\mathcal{M}$, we can integrate the denoising process with both non-deep and deep RS models as follows:
\begin{equation}
\mathcal{L}(\bm{U},\bm{V}) = ||\bm{X} - \bm{U}\bm{V}^{T}||_F^2 + \zeta (||\bm{U}||_2^2 + ||\bm{V}||_2^2),
\label{eq:nondeep}
\end{equation}
\begin{equation}
\mathcal{L}(\bm{U},\bm{V}) = \mathcal{C}(\bm{X},G_y(\bm{U}, \bm{V}))+ \zeta (||\bm{U}||_2^2 + ||\bm{V}||_2^2),
\label{eq:deep}
\end{equation}
where $\bm{U} \in \mathbb{R}^{M \times d}$ and $\bm{V} \in \mathbb{R}^{N \times d}$ represent the user and item embeddings. $G_y$ denotes a classification network. The hyperparameter $\zeta$ is to balance the loss and the regularization term.
In non-deep learning models, Eq.~\ref{eq:nondeep} can be solved effectively through many methods (\eg, Alternative Least Square).
Since the prediction scores represent the user-item similarities where $\mathcal{M}$ has negative values, we calculate the cost matrix as $\mathcal{M} = -\bm{U}\bm{V}^T$ by using inner product. 
In deep learning models, we can optimize Eq.~\ref{eq:deep} by adopting the widely used binary cross-entropy loss (\ie, $\mathcal{C}(\cdot , \cdot)$). Then, similar calculations can be made for non-deep learning ones, where $\mathcal{M} = -G_y(\bm{U},\bm{V})$.

% \noindent\textbf{Computing matching cost matrix $\mathcal{M}$ via deep methods.}
% We first apply an embedding network (\eg, MLP) to compute user embedding as $U$ and item embedding as $V$.
% To address recommendation with implicit feedback as a binary classification problem \cite{he2017neural}, we denote a classification network ($G_y$) for prediction and adopt the widely used binary cross-entropy loss ($\mathcal{C}(\cdot , \cdot)$) to model interactions $\bm{X}$:

% Considering that the prediction scores have negative correlations with modeling cost, we have $\mathcal{M} = -G_y(\bm{U},\bm{V})$.

However, there are three technical challenges when applying OT to denoised recommendation: 
% 1. one user/item can interact with many items/users, which requires a many-to-many matching; 2. the number of users and items makes the matching process time-consuming; 3. how to leverage the learned matching matrix and finally eliminate the effect of noises for denoised recommendation.

\noindent\textbf{Challenge 1. }In RSs, each user can interact with many items and each item interact with many users, which corresponds to a many-to-many interaction matrix. Unlike previous applications of OT in CV or NLP, directly applying the one-hot constraint specified in Eq.~\ref{eq:emd} is counter-intuitive for recommendation.

\noindent\textbf{Challenge 2. }In RSs, the number of users and items in the real world tends to be large, making the matching process time-consuming.

\noindent\textbf{Challenge 3. }Besides distinguishing noises, it is necessary to design a mechanism to eliminate the effect of noisy interactions. A na\"{i}ve way is to set one global thresholding hyperparameter $\kappa$ to keep the top $\kappa$ items with the least matching cost for each user.
However, the global thresholding hyperparameter cannot accommodate individual user behaviors.

% interactions
% Though the proposed Relaxed OT can solve many-to-many optimal transport for recommendation, it fails to eliminate the effect of noisy interactions, for the denoised model still consider all user-item interactions by their matching scores. 
% To better stress the well-matched interactions and reduce the noisy ones, we proposed Partial Relaxed OT and a na\"{i}ve way is to leverage one global thresholding hyperparameter $\kappa$ to keep the top $\kappa$ items with the least matching cost for each user.
% However, a global threshold cannot accommodate individual user behaviors. 
% As shown in Figure \ref{fig:model}, some users with consistent (narrow) interests (\ie, user $A$) would benefit from a reasonably low threshold since a few items can already represent the preference and more items can introduce more noises. Meanwhile, some users with diverse (wide) interests (\ie, user $B$) would require the threshold to be relatively high to model a broader range of preferences and satisfy the need for exploration. %leading to an unresolvable conflict in one hyperparameter. 
% Therefore, a fixed $\kappa$ will limit the flexibility of modeling users' personalities and different information needs.

\subsection{ProRec for Denoised Recommendation}
In this section, we propose a Partial relaxed optimal Transport for denoised Recommendation (ProRec) framework to address the above three challenges.

\noindent\textbf{Many-to-many matching to meet the nature of RSs. }
To address Challenge 1, we look for a smoother version of OT by lowering the sparsity and increasing the entropy of the matching matrix \cite{DBLP:journals/pami/CourtyFTR17}.
Among the many OT algorithms (\eg, ADMM and Proximal Point Method \cite{boyd2011distributed,parikh2014proximal}), we choose the Sinkhorn algorithm \cite{cuturi2013sinkhorn} due to its simplicity and efficiency.
We achieve OT with a many-to-many matching through the Sinkhorn algorithm as follows:
\begin{equation}
\begin{aligned}
    &\min_{\bm{\pi} \in \mathcal{X}(\mathcal{D}^u,\mathcal{D}^v)} \sum_{i=1}^M \sum_{j=1}^N \mathcal{\bm{M}}_{ij}\pi_{ij} + 
\gamma H(\bm{\pi}) \\
    & s.t.\,\, \bm{1}_M\bm{\pi} = \bm{p} , \bm{1}_N\bm{\pi}^T = \bm{q},
\end{aligned}
\label{eq:sinkhorn}
\end{equation}
where $H(\bm{\pi}) = \sum_{i=1}^M \sum_{j=1}^N \pi_{ij}\left( \log (\pi_{ij})-1\right)$. $\bm{1}_M \in \mathbb{R}^{1\times M},\bm{1}_N \in \mathbb{R}^{1\times N}$ are with all elements equal to 1. $\gamma$ is a hyperparameter to balance the entropy regularization and OT. 

Compared with most studies of discrete OT in CV, the major advantages of the Sinkhorn algorithm are as follows:
\begin{itemize}
    \item By relaxing the discrete constraint, we can obtain a continuous matrix instead of a one-hot one. In this way, the obtained matching matrix can capture the many-to-many relationships between users and items.
    \item By adding an entropy regularization, we alleviate the sparsity problem of $\bm{\pi}^*$ and can use dense cost matrix for subsequent ranking.
    \item By modifying the distribution constraint of $\bm{\pi}$ from uniform (\ie, $\bm{1}_M\bm{\pi} = \frac{1}{N}\bm{1}_N$) to popularity-based (\ie, $\bm{1}_M\bm{\pi} = \bm{p}$), we model the users and items that have different number of interactions, and successfully adapt the optimization to the recommendation scenario.
    Specifically, we count the interactions and then normalize them to obtain $\bm{p}$ and $\bm{q}$ for the distributions of items and users, respectively, where $\sum_{j=1}^N \bm{p}_j = 1$ and $\sum_{i=1}^M \bm{q}_i = 1$.
\end{itemize}

\noindent\textbf{Relaxed OT to reduce the time complexity. }Note that the optimal $\bm{\pi}^*$ is dense and needs to be computed through multiple iterations. Therefore, it is time-consuming in RSs with large numbers of users and items.
Inspired by the Relaxed OT for EMD \cite{kusner2015word}, we extend the original Sinkhorn algorithm with a relaxed regularization. Specifically, we use two auxiliary distances, each essentially is the Sinkhorn with only one of the constraints in Eq.~\ref{eq:sinkhorn}:
\begin{equation}
% \begin{aligned}
    \min_{\bm{\pi}} \sum_{i=1}^M \sum_{j=1}^N \mathcal{\bm{M}}_{ij}\pi_{ij} + \gamma H(\bm{\pi}) \quad s.t.\,\,\bm{1}_M\bm{\pi} = \bm{p}, 
% \end{aligned}
\label{eq:p}
\end{equation}
\begin{equation}
% \begin{aligned}
    \min_{\bm{\pi}} \sum_{i=1}^M \sum_{j=1}^N \mathcal{\bm{M}}_{ij}\pi_{ij} + \gamma H(\bm{\pi}) \quad s.t.\,\, \bm{1}_N\bm{\pi}^T = \bm{q}.
% \end{aligned}
\label{eq:q}
\end{equation}
\begin{prop}
The constraint of problem in Eq.~\ref{eq:sinkhorn} is relaxed by Eq.\ref{eq:p} and Eq.~\ref{eq:q}. By defining $\Tilde{\bm{\mathcal{M}}} = e^{-\frac{\mathcal{\bm{M}}}{\gamma}}$, the closed-form solution for the global optimal matching matrix is:
\begin{equation}
    \bm{\bm{\pi}^*} = \max(\Tilde{\mathcal{\bm{M}}} \operatorname{diag}(\bm{p} \oslash \bm{1}_M \Tilde{\mathcal{\bm{M}}}),
    \operatorname{diag}(\bm{q} \oslash \bm{1}_N \Tilde{\mathcal{\bm{M}}}^T) \Tilde{\mathcal{\bm{M}}}),
\label{eq:rot}
\end{equation}
where $\oslash$ corresponds to the element-wise division.
\label{proposition:rot}
\end{prop}
% \begin{proof}
% See Supplementary Materials.
% \end{proof}

\begin{proof}
Based on Eq.~\ref{eq:q} and using the Lagrange multiplier method we have:
\begin{equation}
    \mathcal{L} = \sum_{i=1}^M \sum_{j=1}^N \mathcal{\bm{M}}_{ij}\bm{\pi}_{ij} + \gamma H(\bm{\pi}) - (\bm{1}_N \bm{\pi}^T - \bm{q})\bm{f}^T,
\end{equation}
where $\bm{f}^T$ is Lagrangian multipliers.

Then, in order to find the optimal solution, we let the gradient $\frac{\partial \mathcal{L}}{\partial \bm{\pi}_{ij}} = 0$, which means:
\begin{equation}
    \frac{\partial \mathcal{L}}{\partial \bm{\pi}_{ij}} = \mathcal{\bm{M}}_{ij} + \gamma \log \bm{\pi}_{ij} - \bm{f}_i = 0,
\end{equation}
In this case, we can obtain the solution of $\bm{\pi}_{ij}$ as
\begin{equation}
    \bm{\pi}_{ij} = e^{\frac{f_i-\mathcal{\bm{M}}_{ij}}{\gamma}}.
\label{eq:10}
\end{equation}

The users and items that have different interactions, which means both users and items are based on different degrees of popularity. To take such popularity into consideration, we count and norm the number of items that have interacted with user $i$. Based on the popularity of user $i$ ( \ie, $\bm{q}_i$), we add constraints of the matching matrix $\pi$ by $\sum_{j=1}^{N} \bm{\pi}_{ij} = \bm{q}_i$. Then, we have
\begin{equation}
    \sum_{j=1}^{N} e^{\frac{\bm{f}_i-\mathcal{\bm{M}}_{ij}}{\gamma}} = \bm{q}_i.
\end{equation}

Since the Lagrange multiplier $\bm{f}_i$ is not related to $j$, we can extract this part and rewrite the formula as:
\begin{equation}
    \frac{1}{\bm{q}_i} \sum_{j=1}^{N} e^{-\frac{\mathcal{\bm{M}}_{ij}}{\gamma}} = e^{-\frac{\bm{f}_i}{\gamma}}.
\label{eq:12}
\end{equation}

Take the logarithm of both sides of the above equation and move $\gamma$ to the other side:

\begin{equation}
    f_i = -\gamma \log (\frac{1}{\bm{q}_i} \sum_j e^{-\frac{\mathcal{\bm{M}}_{ij}}{\gamma}} ).
\end{equation}

To keep the formula simple and clear, we define $\Tilde{\bm{\mathcal{M}}} = e^{-\frac{\mathcal{\bm{M}}}{\gamma}}$, so the formula above can be written as:

\begin{equation}
    \bm{f} = \gamma \log{\bm{q}} - \gamma \log (\bm{1}_N \Tilde{\bm{\mathcal{M}}}^T).
\end{equation}

According to Eq.~\ref{eq:10} and Eq.~\ref{eq:12}, we have
\begin{equation}
    \bm{\pi}_{ij} = e^{\frac{\bm{f}_i - \mathcal{\bm{M}}_{ij}}{\gamma}} = \bm{q}_i \frac{e^{-\frac{\mathcal{\bm{M}}_{ij}}{\gamma}}}{\sum_j e^{-\frac{\mathcal{\bm{M}}_{ij}}{\gamma}}}.
\label{eq:15}
\end{equation}

Finally, based on Eq.~\ref{eq:15} and the definition of $\Tilde{\bm{\mathcal{M}}} = e^{-\frac{\mathcal{\bm{M}}}{\gamma}}$, we obtain the closed-form solution via a matrix form as follow
\begin{equation}
\bm{\pi}_V= \operatorname{diag}(\bm{q} \oslash \bm{1}_N \Tilde{\mathcal{\bm{M}}}^T) \Tilde{\mathcal{\bm{M}}},
\end{equation}
where $\oslash$ corresponds to element-wise division.

Following the same derivation, we can obtain the closed-form solution of Eq.~\ref{eq:p} as follow
\begin{equation}
\bm{\pi}_U=\Tilde{\mathcal{\bm{M}}} \operatorname{diag}(\bm{p} \oslash \bm{1}_M \Tilde{\mathcal{\bm{M}}}).
\end{equation}

The original Sinkhorn minimizes $\pi$ under two constraints. By relaxing one constraint, we can balance the computation efficiency and improved preformance.
To approximate the solution under the original constraints, we adopt the $\max$ operation, which is consistent with the results in previous work \cite{kolkin2019style}. So we have
% cost will be smaller and the result will be better now that one constraint is relaxed, the result can only be better, that is, smaller. If we want to approximate the solution under the original constraints, we must make the relaxed solution larger. So we have
\begin{equation}
    \bm{\pi}^* = \max(\bm{\pi}_U, \bm{\pi}_V)
\end{equation}
This is equivalent to:
\begin{equation}
    \bm{\pi}^* = \max(\Tilde{\mathcal{\bm{M}}} \operatorname{diag}(\bm{p} \oslash \bm{1}_M \Tilde{\mathcal{\bm{M}}}),
    \operatorname{diag}(\bm{q} \oslash \bm{1}_N \Tilde{\mathcal{\bm{M}}}^T) \Tilde{\mathcal{\bm{M}}}).
\end{equation}
\end{proof}

Note that, the closed-form solution yields a time complexity of $O$(max$(M, N)^2$). This shows that by relaxing the OT formulation with auxiliary limitation from only one side, our model can achieve a lower time complexity than the $O$(max$(M, N)^3$) reported in \cite{bhushan2018deepjdot}, which uses OT for one-to-one matching.

\begin{algorithm*}
\caption{Partial Relaxed Optimal Transport for Denoised Recommendation (ProRec)}\label{algorithm}
\KwData{The noisy user-item interactions $\bm{X}$, user popularity $\bm{p}$ and item popularity $\bm{q}$}
\KwResult{Denoised user embeddings $\bm{U}$ and item embeddings $\bm{V}$, a relabeled matrix $\bm{R}$.}
% Initialized $\bm{U}$ and $\bm{V}$.

\While{not converged}
{
    \emph{1. Update $\bm{U}$ and $\bm{V}$}\;
  \lIf{Learning $\bm{U}$ and $\bm{V}$ via non-deep methods}
    {Update $\bm{U}$ and $\bm{V}$ via $||\bm{X} - \bm{U}\bm{V}^{T}||_F^2 + \zeta (||\bm{U}||_2^2 + ||\bm{V}||_2^2$ in Eq.~\ref{eq:nondeep} and calculate $\mathcal{M} = -\bm{U}\bm{V}^T$}
    \lElse
    {Update $\bm{U}$ and $\bm{V}$ via $\mathcal{C}(\bm{X},G_y(\bm{U}, \bm{V}))+ \zeta (||\bm{U}||_2^2 + ||\bm{V}||_2^2)$ in Eq.~\ref{eq:deep}, where $\mathcal{C}(\cdot , \cdot)$ is binary cross-entropy loss; Then calculate $\mathcal{M} = -G_y(\bm{U},\bm{V})$, where $G_y$ is a classification network for prediction}
    
  \emph{2. Update $\bm{X}$}\;
  Model denoised recommendation problem by minimizing $\sum_{i=1}^M \sum_{j=1}^N \mathcal{\bm{M}}_{ij}\pi_{ij} + 
\gamma H(\bm{\pi})$, where $\bm{1}_M\bm{\pi} = \bm{p},\bm{1}_N\bm{\pi}^T = \bm{q}$, in Eq.~\ref{eq:sinkhorn}\;
  Compute a global optimal matching matrix $\bm{\pi}^*$ by the closed-form solution $\bm{\bm{\pi}^*} = \max(\Tilde{\mathcal{\bm{M}}} \operatorname{diag}(\bm{p} \oslash \bm{1}_M \Tilde{\mathcal{\bm{M}}}),
    \operatorname{diag}(\bm{q} \oslash \bm{1}_N \Tilde{\mathcal{\bm{M}}}^T) \Tilde{\mathcal{\bm{M}}})$ in Eq.~\ref{eq:rot}\;
  Compute personalized thresholds $\bm{\kappa}$ by automatic splitting mechanism in Eq.~\ref{eq:kappa}\;
  Relabel user-item interaction by sorting $\bm{\kappa}$ in Eq.~\ref{eq:relabel}\;
  Update $\bm{X}$ by $\bm{X} = \lambda \bm{X} + (1-\lambda) \bm{R} \odot \bm{X}$ in Eq.~\ref{eq:newx}\;
}
\end{algorithm*}

\noindent\textbf{Personalized thresholding mechanism for denoised recommendation. }
To flexibly discriminate the intrinsic and noisy user-item interactions for individual users, we propose a non-parametric personalized thresholding mechanism.
We first normalize the matching cost by each user as ${\pi_{ij}}^{*\prime}=\frac{\pi^*_{ij}}{\sum_j \pi^*_{ij}}$, and rank them as $\bm{\rho_i} = {\rm sort} (\bm{\pi}^{*\prime}_{i})$. 
According to the definition of the matching matrix, each row represents users' preferences towards the items. We define $\bm{\kappa} = \{\kappa_1, \dots, \kappa_M\}$ as the index of the threshold which can filter out users' noisy preferences. 
$\kappa_{i}$ denotes user $i$'s threshold.
As shown in Figure \ref{fig:model}, some users with consistent (narrow) interests (\ie, user $A$) would benefit from a reasonably low threshold since a few items can already represent their preferences and more items can introduce more noises. Meanwhile, some users with diverse (wide) interests (\ie, user $B$) would require the threshold to be relatively high to model a broader range of preferences and satisfy the need for exploration.
To find a personalized splitting point $\kappa_{i}$ for each user according to the learned matrix, the model requires an automatic splitting mechanism. Inspired by the threshold selection mechanism used in CART \cite{angelopoulos2005exploiting}, we efficiently learn the personalized thresholds $\kappa_{i}$ for different users by minimizing the following sum of square errors:
\begin{equation}
\kappa_{i} = \arg\min_{\eta} \left[ \sum_{j=1}^{\eta}\left(\rho_{ij}-c^{(1)}_{\eta}\right)^{2}+\sum_{j=\eta+1}^{N}\left(\rho_{ij}-c^{(2)}_{\eta}\right)^{2}\right],
\label{eq:kappa}
\end{equation}
where $\rho_{ij}$ is the $j$-th value of the sorted $\rho_i$ and the index $\eta \in \{1,2,\cdots,N-1\}$. 
$c^{(1)}_{\eta}$ represents the mean square error of the top $\eta$ items while $c^{(2)}_{\eta}$ represents one of the rest $N-\eta$ items. We have
\begin{equation}
\begin{aligned}
c^{(1)}_{\eta}=\frac{1}{\eta}\sum_{j=1}^{\eta}\rho_{ij}, \quad c^{(2)}_{\eta}=\frac{1}{N-\eta}\sum_{j=\eta+1}^{N}\rho_{ij}.
\end{aligned}
\end{equation}
After obtaining the index of splitting points $\kappa$, we have the corresponding threshold values $\rho_{i, \kappa_i}$.
Based on the personalized threshold $\rho_{i, \kappa_i}$, we can relabel by $r_{ij} \in \mathbf{R}$ as follows:
\begin{equation}
r_{ij} = \frac{1}{1+\exp(-\beta (\rho_{ij}-\rho_{i, \kappa_i}))},
%\omega_{ij} = 1-\frac{1}{2}(1 + {\rm sgn}(\rho_{ij}-\kappa_{i,z}))
\label{eq:relabel}
\end{equation}
where $r_{ij}$ is a monotonically increasing function according to the value of $\rho_{ij}-\rho_{i, \kappa_i}$. $\beta$ is a hyperparameter to control the inclination of the function. 

Finally, to make the relabeling process more flexible, we propose to use a hyperparameter $\lambda$ to control the degree of relabeling and obtain a new interaction matrix as follows:
\begin{equation}
\bm{X} = \lambda \bm{X} + (1-\lambda) \bm{R} \odot \bm{X},
\label{eq:newx}
\end{equation}
where $\odot$ corresponds to the element-wise product. Considering that the noises in sparse positive interactions would induce worse effects in recommendation \cite{wang2020denoising}, in this work, we only target the part of the original dataset where $X_{ij} = 1$. In this way, the newly generated $\bm{X}$ can keep users' intrinsic preferences and reduce the noise ones. We summarize the learning procedure of our proposed ProRec in Algorithm \ref{algorithm}, which is proved to guarantee the local convergence for denoised recommendation. 

% \begin{figure}
%     \centering
%     \includegraphics[width=0.75\linewidth]{image/loss.pdf}
%     \caption{The proposed model DeepPROTON for recommendation. The partial optimal transport loss $\mathcal{L}_{ot}$ is adopted to adjust the MF and then enhance the recomendation performance. The cross entropy loss $\mathcal{L}_{e}$ is for local modeling. Note that the $G_e$ denotes the strong matching network to learning embeddings for users and items. The $G_y$ denotes the output classification layers. Best viewed in color.}
%     \label{fig:process}
% \end{figure}

% Considering that the user-item interactions consist of intrinsic part and noisy part. Algorithm 1 

\begin{prop}
Algorithm 1 monotonically decreases the global matching cost between user embedding space $\mathcal{D}^u$ and item embedding space $\mathcal{D}^v$ until local convergence.
\end{prop}
% See Supplementary Materials.
% We sketch the proof as follows: 
% Updating the relabeled matrix $\bm{R}$ can make the new generated interactions $\bm{X}$ closer to the negative cost matrix $-\bm{\mathcal{M}}$ than original $\bm{X}$.
% Updating $\bm{U}$ and $\bm{V}$ based on the new $\bm{X}$ by Eq.~\ref{eq:nondeep} or Eq.~\ref{eq:deep} can further decrease the global matching cost $\bm{\mathcal{M}}$. 
% Detailed proofs of Proposition~2 can be found in Appendix~A.2 of the supplementary material.
\begin{proof} 
In order to prove the monotonically decreasing property of the loss function in the iterative process, we only need to prove the following two inequations:

\begin{equation}
\mathcal{L}(\bm{X}^{(t)},\bm{U}^{(t+1)},\bm{V}^{(t+1)}) \leq \mathcal{L}(\bm{X}^{(t)},\bm{U}^{(t)},\bm{V}^{(t)}),
\label{eq:fixX}
\end{equation}

\begin{equation}
\mathcal{L}(\bm{X}^{(t+1)},\bm{U}^{(t+1)},\bm{V}^{(t+1)}) \leq \mathcal{L}(\bm{X}^{(t)},\bm{U}^{(t+1)},\bm{V}^{(t+1)}),
\label{eq:fixUV}
\end{equation}

where:
\begin{equation}
\mathcal{L}(\bm{X},\bm{U},\bm{V}) = ||\bm{X} - \bm{U}\bm{V}^{T}||_F^2 + \zeta (||\bm{U}||_2^2 + ||\bm{V}||_2^2)
\end{equation}

It is obvious that Eq.~\ref{eq:fixX} holds since we just simply choose $\bm{U}$ and $\bm{V}$ to minimize it in each iteration. Now we prove Eq.~\ref{eq:fixUV} holds:

At iteration $t$, according to the update rules Eq.~\ref{eq:newx} we have:

\begin{equation}
\bm{X}^{(t+1)}=\lambda\bm{X}^{(t)}+(1-\lambda)\bm{R}^{(t)}\odot\bm{X}^{(t)}.
\end{equation}

Since $\bm{U}$ and $\bm{V}$ are fixed, the regularization term in the loss function will not change when we just update the interaction matrix. Thus, proving Eq.~\ref{eq:fixUV} is equivalent to proving the following inequation:

\begin{equation}
\begin{aligned}
&||\lambda\bm{X}^{(t+1)}+(1-\lambda)\bm{R}^{(t)}\odot\bm{X}^{(t)} + \bm{\mathcal{M}}^{(t)}||_F^2 \\
\leq &||\bm{X}^{(t)} + \bm{\mathcal{M}}^{(t+1)}||_F^2.
\end{aligned}
\end{equation}

For every user-item pair $(i,j)$, $i\in\{1,2,...,M\},j\in\{1,2,...,N\}$, if $x_{ij}$ is a noise sample (User's wrong click), when we update $\bm{U},\bm{V}$ to minimize the loss, $\mathcal{M}_{ij}$ increases. According to Eq.~\ref{eq:rot}, the elements in the matching matrix can be written as:

\begin{equation}
\pi^{*}_{ij}=e^{-\frac{\mathcal{M}_{ij}}{\gamma}}\frac{ \sum_{k=1}^M \pi_{kj}}{\sum_{k=1}^M e^{-\frac{\mathcal{M}_{kj}}{\gamma}}}
\end{equation}

We can see that when $\mathcal{M}_{ij}$ increases, the value of $\pi^{*}_{ij}$ will decrease, which means that the preference ranking of item $j$ for user $i$ will drop. According to the relabel process Eq.~\ref{eq:relabel}, a lower rank in the item list means a lower value of $r_{ij}$. In other words, the following equation holds for all user-item pairs:

\begin{equation}
(\lambda x_{ij}^{(t)}+(1-\lambda) r_{ij}^{(t)} x_{ij}^{(t)}+\mathcal{M}_{ij}^{(t+1)})^2 \leq (x_{ij}^{(t)}+\mathcal{M}_{ij}^{(t+1)})^2,
\end{equation}
where $0\leq\lambda\leq 1$ and the equal sign holds when $\lambda=1$. From the definition of the F-norm, we can derive that:

\begin{equation}
\begin{aligned}
||\lambda\bm{X}^{(t)}+&(1-\lambda)\bm{R}^{(t)}\odot\bm{X}^{(t)}+\bm{\mathcal{M}}^{(t+1)}||_F^2\\
& \leq ||\bm{X}^{(t)}+\bm{\mathcal{M}}^{(t+1)}||_F^2
\end{aligned}
\end{equation}

The equation means that compare with matrix $\bm{X}^{(t)}$, $\bm{R}^{(t)}\odot\bm{X}^{(t)}$ is closer to $\bm{\mathcal{M}}^{(t+1)}$ in the matrix space. Therefore, Eq.~\ref{eq:fixUV} holds for each iteration. The proof above is similar when using cross-entropy as loss function.

\end{proof}
\section{Experiments}

In this section, we first conduct controlled experiments with synthetic noises on one dataset, so as to investigate ProRec's performance to different levels of noisy interactions. Then, we evaluate ProRec's performance on three original real-world datasets of recommendation with implicit feedback.
To evaluate our proposed ProRec framework, we mainly focus on the following three research questions.
\begin{itemize}
  \item\textbf{RQ1: }Can the proposed ProRec distinguish noise?
  \item\textbf{RQ2: }How does the proposed ProRec perform in comparison to state-of-the-art recommendation methods?
  \item\textbf{RQ3: }What are the effects of different model components and different hyperparameters?
  \item\textbf{RQ4: }How does ProRec improve recommendation by distinguishing intrinsic and noisy interactions? 
\end{itemize}
% we evaluate our proposed ProRec framework focusing on the following three research questions: \textbf{RQ1: }How does the proposed ProRec perform in comparison to state-of-the-art recommendation methods? \textbf{RQ2: }What are the effects of different model components and different hyperparameters? \textbf{RQ3: }How does ProRec improve recommendation by distinguishing intrinsic and noisy interactions? 
\begin{table*}
  \caption{Statistics of the datasets used in our experiments.}
  \normalsize
  \centering
  \begin{tabular}{c c cccc}
      \toprule
      Category & Dataset & \# \ User & \# \ Item & \# \ Interaction & Sparsity \\
      \hline
      Controlled & ML-100k & 942 & 1,447 & 55,375 & 0.9594 \\
      \hline
      \multirow{3}{*}{Original} & AMusic & 1,776 & 12,929 & 46,087 & 0.9980 \\

      & AToy & 3,137 & 33,953 & 84,642 & 0.9992 \\

    %   ML20M & 138,493 & 27,278 & 12,195,566 & 0.9965 \\
    %   \hline
    %   Netflix & 2,649,430 & 17,771 & 56,919,190 & 0.9988 \\
    %   \hline
      & Yahoo & 1,948,882 & 46,110 & 48,817,561 & 0.9995 \\
      \bottomrule
  \end{tabular}
\label{tab:datasets-stat}
\end{table*}

\begin{figure*}
\centering
  \subfloat[HR/Recall, controlled dataset]{
    \includegraphics[width = 0.35\linewidth]{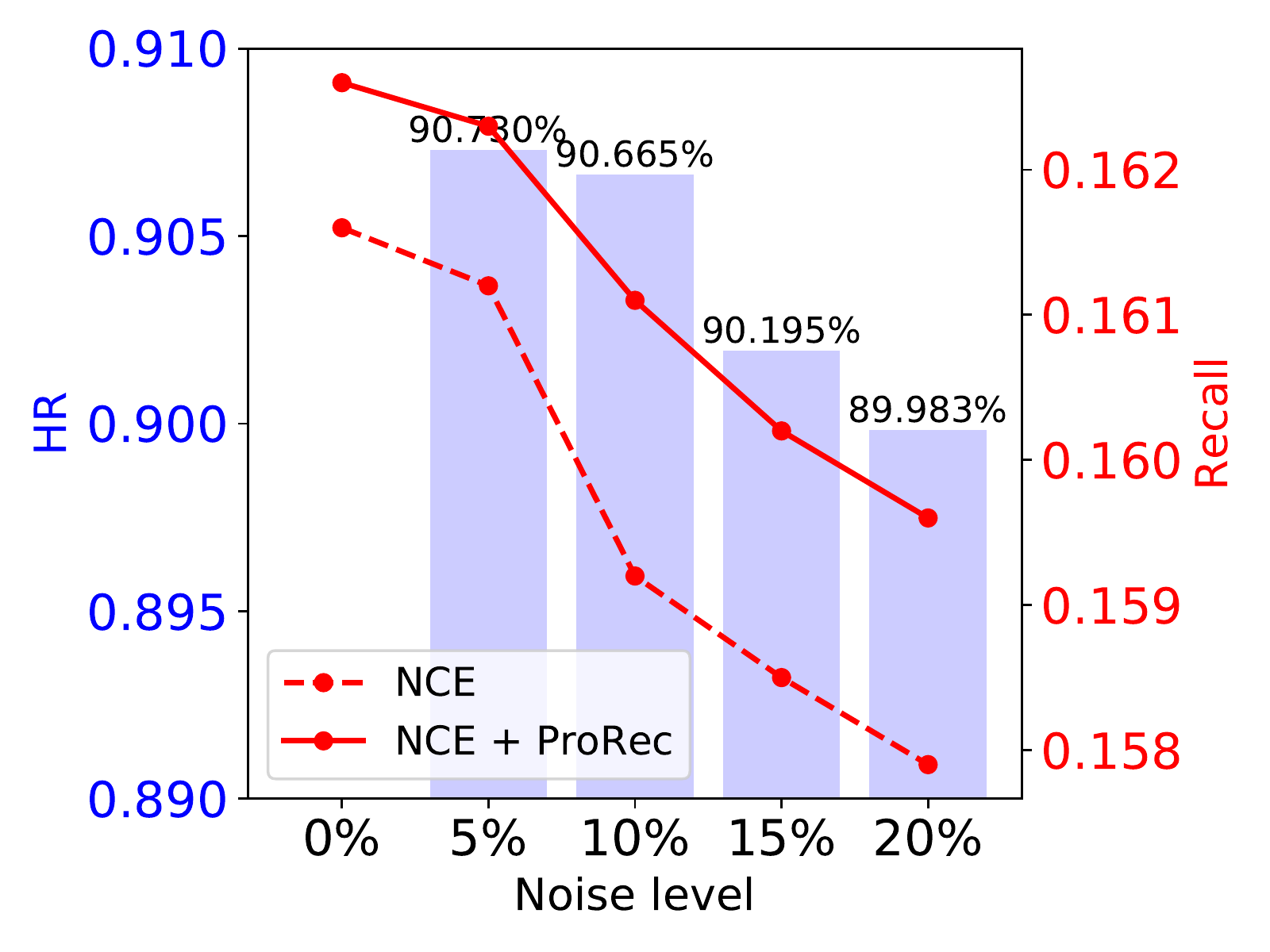}
    \label{fig:syn}
    }
  \hfill
  \subfloat[Varying $\lambda$ on AMusic]{
    \includegraphics[width = 0.3\linewidth]{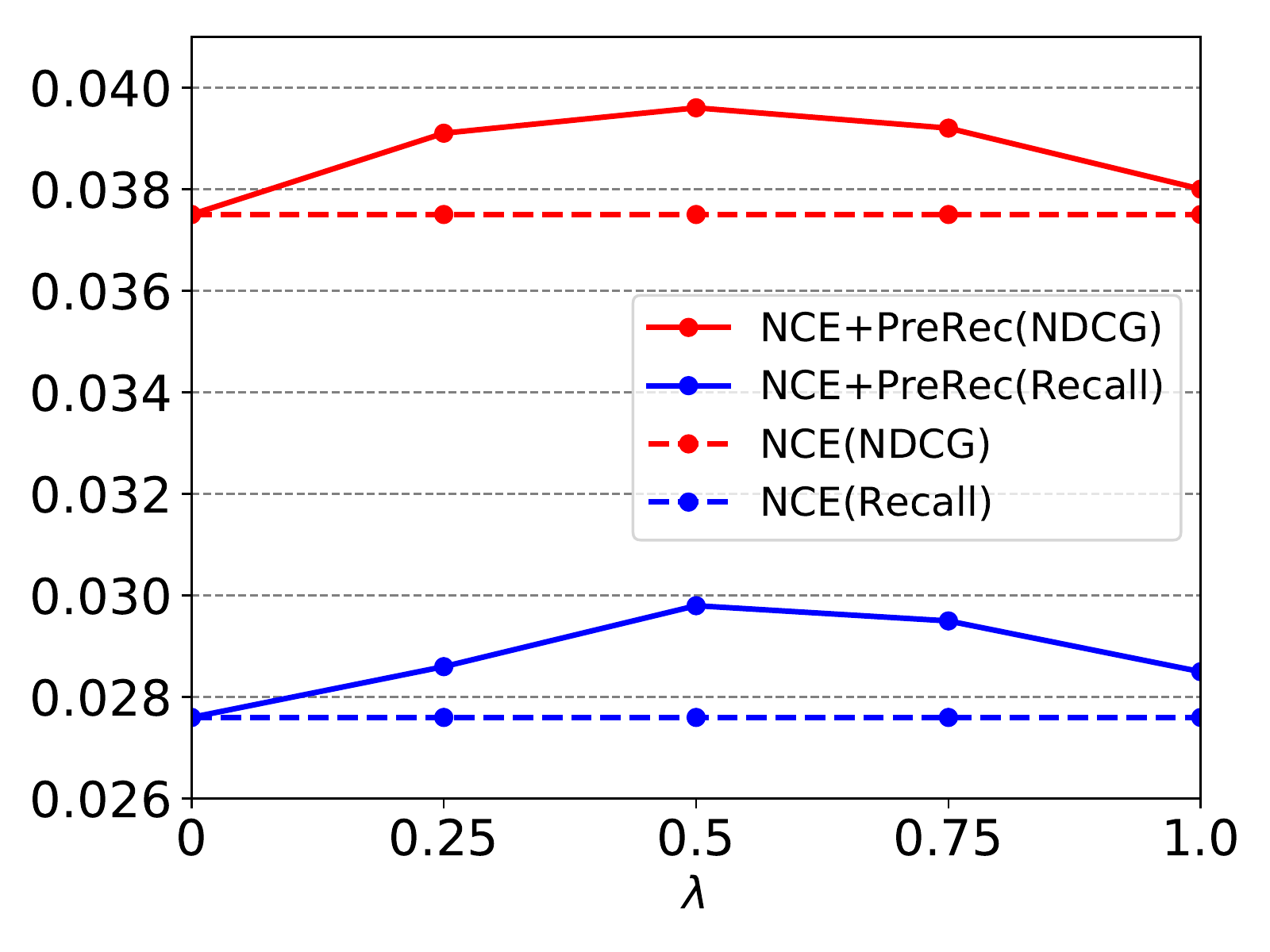}
    \label{fig:lambda}
  }
  \hfill
  \subfloat[Varying $\gamma$ on AMusic]{
    \includegraphics[width =0.3\linewidth]{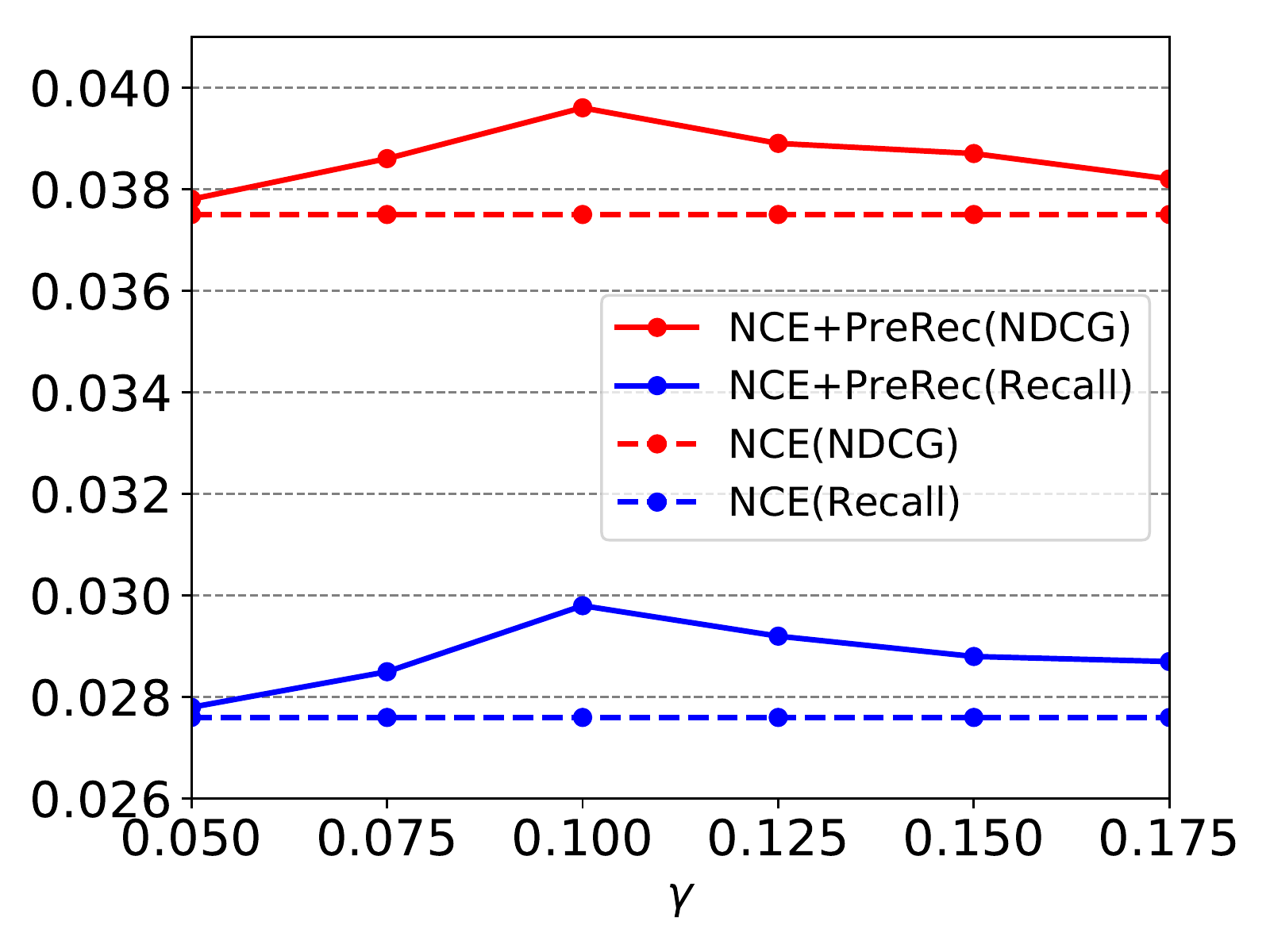}
    \label{fig:gamma}
  }
  \caption{(a) Testing HR/Recall with different noise level on the noise-controlled dataset. (b)-(c) Hyperparameter effect on the proposed ProRec.}
\end{figure*}

\subsection{Experimental Settings}
% \noindent\textbf{Dataset.}
\subsubsection{Dataset and evaluation protocols. }We use ML-100k for the noise-controlled experiments and do a 4:1 random splitting for train/test data following~\cite{ding2020simplify}. To simulate the noise, we randomly select 5\%/10\%/15\%/20\% of ground-truth records and flip labels in the train set of ML-100k. The selected records can be regarded as noises during training. 
As for real data experiments, we use Amazon music (AMusic)\textsuperscript{\ref{data}}, Amazon toys (AToy)\footnote{https://github.com/familyld/DeepCF/tree/master/Data\label{data}}, and Yahoo\footnote{https://webscope.sandbox.yahoo.com/catalog.php?datatype=r}.
These datasets have been widely adopted in previous literature~\cite{deng2019deepcf,wang2020denoising,wu2019noise}, and their statistics are summarized in Table~\ref{tab:datasets-stat}.
We do a 5:2:3 splitting for them following~\cite{wu2019noise}. 
We evaluate the ability of distinguishing noise by Hit Ratio (HR), and recommendation performances by normalized discounted cumulative (NDCG), mean average precision (MAP), and Recall.

\subsubsection{Methods for comparison }
We first compare two types of algorithms for recommendation.

The first type is non-deep algorithms, which include:
\begin{itemize}
    \item SVD~\cite{cremonesi2010performance}: A similarity-based recommendation method that constructs a similarity matrix through SVD decomposition of implicit matrix $\bm{X}$; 
    \item NCE~\cite{wu2019noise}: Noise contrastive estimation item embeddings combined with a personalized user model based on PLRec~\cite{sedhain2016practical}. 
\end{itemize}

The second type is deep learning methods, including:
\begin{itemize}
    \item CDAE~\cite{hu2008collaborative}: Collaborative Denoising Autoencoder is a deep learning method specifically optimized for implicit feedback recommendation tasks; 
    \item WRMF~\cite{hu2008collaborative}: Weighted Regularized Matrix Factorization, which is a deep MF framework; VAE-CF~\cite{liang2018variational}: Variational Autoencoder for Collaborative Filtering, which is one of the state-of-the-art deep learning based recommendation algorithm~\cite{wu2019noise}. 
\end{itemize}

Besides the above methods, R-CE and T-CE are the only existing frameworks for denoised recommendation proposed by~\cite{wang2020denoising}, which are designed for deep learning methods only:
\begin{itemize}
\item CDAE + T-CE~\cite{wang2020denoising}: Truncated Loss that discards the large-loss samples with a dynamic threshold in each iteration based on CDAE. According to the original paper~\cite{wang2020denoising}, CDAE + T-CE shows the best performance compared with the others based on GMF and NeuMF. Therefore, we add CDAE + T-CE as baselines to compare the ability of denoising with the proposed ProRec.
\item CDAE+R-CE \cite{wang2020denoising}: Reweighted Loss that lowers the weight of large-loss samples based on CDAE.
\end{itemize}

For all the algorithms above (except for CDAE + T-CE and CDAE + R-CE), we add the proposed ProRec on top of them, which allows both non-deep and deep RS models to discriminate intrinsic and noisy interactions in an unsupervised fashion. Specifically, we have SVD + ProRec, NCE + ProRec, CDAE + ProRec, WRMF + ProRec, and VAE-CF + ProRec, respectively. 

\subsubsection{Implementation details}
% Interact and decide hyperparameter setting
We implement the proposed ProRec with Pytorch on 8 NVIDIA 2080Ti GPUs. The code will be released publicly after acceptance. Implementations of the compared baselines are either from open-source projects or the original authors (SVD/NCE/CDAE//WRMF/VAE-CF\footnote{https://github.com/wuga214/NCE\_Projected\_LRec}, T-CE and R-CE\footnote{https://github.com/WenjieWWJ/DenoisingRec}). 
We optimize ProRec with standard Adam. We tune all hyperparameters through grid search.
In particular, learning rate in 
$\{0.0005, 0.001, 0.005, 0.01\}$, 
$\lambda$ in $\{0.25, 0.5, 0.75, 1.0\}$,
$\gamma$ in $\{0.05, 0.075, 0.1, 0.125, 0.15, 0.175\}$, 
$\beta$ in $\{1, 5, 10, 20, 50\}$,
$\zeta$ in $\{0.0005, 0.001, 0.005, 0.01\}$,
and the batch size in $\{100, 250, 500, 1000\}$ for different datasets.
We also carefully tuned the hyperparameters of all baselines through cross-validation as suggested in the original papers to achieve their best performances.

% \noindent\textbf{Implementation details.}
% We evaluate the recommendation performance using three metrics: NDCG@K, MAP@K, Recall@K instead of sampled metrics as suggested in~\cite{krichene2020sampled}.
% Implementations of the compared baselines are either from open-source projects or the original authors (SVD/CDAE/VAE-CF/WRMF/NCE\footnote{https://github.com/wuga214/NCE\_Projected\_LRec}, T-CE and R-CE\footnote{https://github.com/WenjieWWJ/DenoisingRec}). 
% We optimize ProRec with standard Adam. We tune all hyperparameters through grid search.
% In particular, learning rate in 
% $\{0.0005, 0.001, 0.005, 0.01\}$,  
% $\gamma$ in $\{0.01, 0.1, 1.0, 10\}$, 
% $\beta$ in $\{1, 5, 10, 20, 50\}$,
% $\zeta$ in $\{0.0005, 0.001, 0.005, 0.01\}$,
% $\lambda$ in $\{0.25, 0.5, 0.75, 1.0\}$,
% and the batch size in $\{100, 250, 500, 1000\}$ for different datasets.  
% We also carefully tuned the hyperparameters of all baselines through cross-validation as suggested in the original papers to achieve their best performances.

\subsection{Controlled Experiments with Synthetic Noises \textbf{(RQ1)}}
Before looking at the performance of recommendation, we first inspect ProRec's ability to distinguish noises.
Since we know the index of the randomly added noises in the dataset, we evaluate the level of distinguished noises by Hit Ratio (HR) of the real added noises in Figure \ref{fig:syn}. 
Specifically, we observe a consistent high HR of NCE + ProRec around 90\% when the level of noises increases.
Furthermore, we test Recall of NCE and NCE + ProRec under different levels of noises.
It can be clearly observed that NCE + ProRec always performs better than NCE, as the former can accurately relabel some noises to eliminate the effect of noises. For example, when the noise level is at 20\%, NCE + ProRec significantly improves NCE from 0.1579 to 0.1596.

\begin{table*}[]
    \caption{Comparison between the original baselines and the ones plus ProRec on three benchmark datasets. Best performances are in boldface. Statistically tested significant improvements with $p$-value < 0.01 brought by ProRec compared with the original methods are indicated with *.}
    % \tabcolsep=0.11cm
    % \small
    \normalsize
    \centering
    \begin{tabular}{l  lll | lll | lll}
   \toprule
    model & N@5 & M@5 & R@5 & 
     N@5 & M@5 & R@5 & N@5 & M@5 & R@5\\
    & \multicolumn{3}{c|}{AMusic} & \multicolumn{3}{c|}{AToy} & \multicolumn{3}{c}{Yahoo}\\
    \hline
    
    % POP & 0.024±0.0027 & 0.0106±0.0013 & 0.0381±0.004  \\
    
    SVD~\cite{cremonesi2010performance} & 0.0363 & 0.0667 & 0.0264 & 0.0116 & 0.0212& 0.0084&
     0.1614& 0.2316& 0.1430 \\
    
    NCE~\cite{wu2019noise} & 0.0375 & 0.0693 & 0.0276 &
    0.0143 & 0.0262 & 0.0104  &
    0.2498 & 0.3440 & 0.2249 \\ 
    
    CDAE~\cite{hu2008collaborative} & 0.0074 & 0.0124 & 0.0060 &
    0.0046 & 0.0080 & 0.0033 &
    0.0716 & 0.1022 & 0.0652\\
    
    % PLRec & 0.0747±0.0059 & 0.0365±0.0034 & 0.1053±0.0077\\
    WRMF~\cite{hu2008collaborative} & 0.0236 & 0.0435 & 0.0171 &
    0.0086 & 0.0157 & 0.0062 &
    0.2503 & 0.3244 & 0.2321\\
    
    VAE-CF~\cite{liang2018variational} &  0.0146 & 0.0258 & 0.0115 &
     0.0117& 0.0212 & 0.0084 & 
    0.2622 & 0.3423 & 0.2380\\

    \hline
    
    CDAE + T-CE~\cite{wang2020denoising} & 0.0070 & 0.0113 & 0.0060 & 0.0048 & 0.0078 & 0.0036 & 0.0696
    & 0.1030 & 0.0607\\
    
    CDAE + R-CE~\cite{wang2020denoising} & 0.0078 & 0.0142 & 0.0057 & 0.0046 & 0.0077 & 0.0034 & 0.0716 & 0.1020 & 0.6490\\ 
    
    \hline
    SVD + ProRec  &  0.0374$^*$ & 0.0677 & 0.0275$^*$  & 
    0.0123$^*$ & 0.0220$^*$ & 0.0093$^*$ &
     0.1987$^*$ & 0.2820$^*$ & 0.1763$^*$\\
    
    NCE + ProRec &  \textbf{0.0396}$^*$ & \textbf{0.0723}$^*$ & \textbf{0.0298}$^*$ &
    \textbf{0.0149} & \textbf{0.0267}$^*$ & \textbf{0.0108} &
    0.2619$^*$ & 0.3562$^*$ & 0.2364$^*$\\
    
    CDAE + ProRec & 0.0080$^*$ & 0.0148$^*$ & 0.0065$^*$ &
    0.0049 & 0.0084 & 0.0038$^*$ &
    0.0718 & 0.1028 & 0.0660$^*$ \\
    
    WRMF + ProRec & 0.0241$^*$ & 0.0442$^*$ & 0.0179$^*$ &
    0.0096$^*$ & 0.0165$^*$ & 0.0077$^*$ &
    0.2699$^*$ & 0.3358$^*$ & 0.2403$^*$ \\ 
    
    VAE-CF + ProRec & 
    0.0162$^*$& 0.0266$^*$& 0.0120&
    0.0129$^*$ & 0.0231$^*$ & 0.0102$^*$ &
    \textbf{0.2743}$^*$ & \textbf{0.3535}$^*$ & \textbf{0.2472}$^*$\\
    % model &  NDCG@50 & MAP@50 & Recall@50 & NDCG@50 & MAP@50 & Recall@50 \\
    % \hline
    % \hline
    % & \multicolumn{3}{c}{ML20M} & \multicolumn{3}{c}{Yahoo}\\
    % \hline
    % % POP & 0.1194±0.0007 & 0.0739±0.0005 & 0.167±0.0011 &
    % % 0.1843±0.0005 & 0.0652±0.0005 & 0.3322±0.0009  \\
    
    % BPR & 0.1540±0.0008 & 0.0927±0.0006 & 0.2267±0.0012 &
    % 0.288±0.0006 & 0.1146±0.0006 & 0.4883±0.0010\\
    
    % CDAE & 0.1536±0.0008 & 0.0932±0.0006 & 0.2219±0.0012 &
    % 0.1843±0.0005 & 0.0652±0.0005 & 0.3323±0.0009\\
    
    % % PLRec & 0.1785±0.0008 & 0.1042±0.0006 & 0.2652±0.0012 &
    % % 0.3673±0.0008 & 0.1779±0.0007 & 0.5363±0.001\\
    % VAE-CF & \textbf{0.1973±0.0009} & 0.1075±0.0006 & \textbf{0.3067±0.0014} &
    % {0.4698±0.0008} & 0.1849±0.0005 & 0.6824±0.0009\\
    
    % WRMF & 0.1962±0.0008 & 0.1106±0.0006 & 0.3021±0.0014 &
    % 0.4633±0.0008 & 0.1795±0.0005 & {0.6881±0.0009} \\
    
    % NCE & 0.1957±0.0009 & {0.1113±0.0006} & 0.2971±0.0014 &
    % 0.4595±0.0008 & {0.1841±0.0005} & 0.6757±0.0009 \\ 
    
    % \hline
    
    % CDAE+T-CE & 0.1531±0.0008 & 0.0937±0.0006 & 0.2191±0.0012 &
    % 0.1828±0.0005 & 0.0674±0.0005 & 0.3290±0.0009\\
    
    % CDAE+R-CE & 0.1468±0.0008 & 0.0893±0.0006 & 0.2118±0.0012 &
    % 0.1843±0.0005 & 0.0652±0.0005 & 0.3323±0.0009\\ 
    
    % \hline
    
    % CDAE+ProRec & 0.1540±0.0008 & 0.0938±0.0006 & 0.2224±0.0012 &
    % 0.1849±0.0005 & 0.0658±0.0005 & 0.3330±0.0009\\
    
    % WRMF+ProRec & 0.1965±0.0008 & 0.1108±0.0006 & 0.3023±0.0014&
    % 0.4654±0.0008 & 0.1807±0.0005 & \textbf{0.6894±0.0009}\\ 
    
    % NCE + ProRec & 0.1959±0.0009 & \textbf{0.1115±0.0006} & 0.2972±0.0014 &
    % \textbf{0.4707±0.0008} & \textbf{0.1884±0.0005} & 0.6824±0.0009\\
     \bottomrule
    \end{tabular}
    \label{tab:overall}
\end{table*}

\subsection{Overall Performance Comparison \textbf{(RQ2)}}
\label{exp:overall}
We compare the recommendation results of the proposed ProRec to those of the baseline models.
Table \ref{tab:overall} shows the NDCG, MAP, and Recall scores on three datasets. We have the following observations.

% \subsubsection{Achieving state-of-the-art}
In three different domains, by integrating ProRec with both non-deep and deep methods, we can consistently improve strong baseline models and achieve the start-of-the-art performance. 
Since deep learning based models are more suitable than non-deep ones in large-scale datasets and vice versa, the proposed NCE + ProRec on top of the non-deep method outperforms all baselines on AMusic and AToy, ranging from 4.33\% (achieved on AMusic on MAP@5) to 7.70\% (achieved on AMusic on Recall@5) and from 1.90\% (achieved on AToy on MAP@5) to 4.20\% (achieved on AToy on NDCG@5), while VAE-CF + ProRec on top of the deep learning method achieve the best performance on the Yahoo dataset, ranging from 3.27\% (achieved on Yahoo on MAP@5) to 4.61\% (achieved on Yahoo on NDCG@5).
% , it is important to integrate our ProRec with both non-deep and deep learning framework, so as to ensure the performance. 
% Take NDCG@5 on the Yahoo dataset for example. VAE-CF + ProRec (0.2743) outperforms the second-best performance NCE (0.4595) before applying ProRec. However, NCE + ProRec (0.4707) outperforms NCE by 2.44\% and finally exceeds VAE-CF.
% performance gain over NCE.
All of these show that ProRec is capable of effective recommendation on different datasets.

% \subsubsection{Improving existing models}
One step further, we observe that the methods based on the proposed ProRec framework all outperform the original models, which indicates the advantage of the denoising process. 
Compared with existing denoising frameworks (\ie, T-CE and R-CE), ProRec can not only be flexibly integrated with both deep and non-deep RS methods, but also gain more improvements.

Moreover, on top of existing RS methods, our proposed ProRec in Eq.~\ref{eq:sinkhorn} does not introduce significant computations.
In the third and sixth columns in Table \ref{tab:ablation}, we compare the model's training time under the same experimental setting (\eg, embedding dimension).
As can be observed, the training time of the proposed NCE + ProRec is within the same scale as NCE on different scales of datasets.

Note that the improvements bringing by the denoising process on the denser datasets (\eg, ML-100k in Figure~\ref{fig:syn} when noise level is 0\%) is less than those on the sparse datasets (\eg, Yahoo in Table~\ref{tab:overall}). 
% With many interactions in the observed data, the RSs based on collaborative filtering have considerable information for implicit denoising. 
Although we can explicitly eliminate some effect of noisy interactions (as shown in Figure~\ref{fig:syn}), they only have a minor impact on the learned embeddings when interactions is enough in the observed data, which is consistent with the results in recent studies~\cite{pal2020non}.

\subsection{Model Ablation and Hyperparameter Study \textbf{(RQ3)}}
\label{sec:rq2}
In this section, three groups of ablation studies are conducted to evaluate the impacts of our proposed many-to-many matching, relaxed regularization, and personalized thresholding mechanisms, as well as the effects of hyperparameters.

\begin{table*}[]
    % \small
    \normalsize
    \centering
    \caption{Impact of different model components of ProRec on AMusic and Yahoo datasets. }
    \begin{tabular}{lcccccc}
    \toprule
    Dataset & \multicolumn{3}{c}{AMusic} & \multicolumn{3}{c}{Yahoo} \\
    Metric& NDCG@5 & Recall@5 & Time & NDCG@5 & Recall@5 & Time \\
    \hline
    NCE &  0.0375 & 0.0276 & 64s & 0.2498 & 0.2249 & 361s\\
    NCE + EMD & 0.0360 & 0.0259 & 98s & 0.2169 & 0.1933 & 723s\\
    NCE + Sinkhorn & 0.0384 & 0.0281 & 132s & 0.2579 & 0.2341 & 1845s\\
    NCE + ROT & 0.0383 & 0.0281 & 76s & 0.2577 & 0.2340 & 683s\\
    NCE + ProRec & \textbf{0.0396} & \textbf{0.0298} & 82s & 
    \textbf{0.2619} & \textbf{0.2364} & 694s\\
    \bottomrule
    \end{tabular}
    \label{tab:ablation}
\end{table*}
% \subsubsection
\noindent\textbf{The effect of continuous many-to-many matching.}
Compared with the many-to-many matching of NCE + Sinkhorn, it is ineffective to directly adopt one-to-one matching (\ie, NCE + EMD) due to the removal of many matched interactions. 
For example, NCE + Sinkhorn significantly improves NCE from 0.0276 to 0.0281 (on Recall@5 on AMusic) while both metrics of NCE + EMD are lower than those of the original NCE.

\noindent\textbf{The efficiency of relaxed regularization.} By balancing the performance and the runtime via relaxing the optimized constraints, we can clearly observe that  NCE + ROT can reduce around half of the runtime while keeping almost the same NDCG and Recall in AMusic and Yahoo datasets.

\noindent\textbf{The effect of personalized thresholding mechanism. }We first use one global threshold for ROT to study the effectiveness of the personalized thresholding mechanism.
In the experiment, we search the hyperparameter of global thresholds $\sigma$ in $\{5, 10, 15, 20, 25\}$ and find that NCE + ROT achieve the best performance when $\sigma=10$.
Compared with  NCE + ROT, NCE + ProRec can individually discriminate intrinsic preferences and noises without tuning the hyperparameter of threshold, and then improve the performance of the partial OT framework.
% For example, the performance gains of NCE + ProRec over NCE + ROT follow the same trend on different datasets across all metrics.
For example, the performance of NCE + ProRec achieves 0.0298, which outperforms 0.0281 of NCE + ROT on Recall@5 on AMusic dataset.

% \subsubsection{Effect of hyperparameter}
\noindent\textbf{The effect of hyperparameters. }Our proposed ProRec framework introduces four hyperparameters: $\gamma$, $\beta$, $\zeta$ and $\lambda$. Since different $\beta$ and $\zeta$ have little effect on performances, we empirically set $\beta$ to 20 and $\zeta$ to 0.001. 
$\lambda$ controls the degree of the relabeling and $\gamma$ controls the sparsity of the matching matrix. 
% As shown in Figure \ref{fig:lambda}-\ref{fig:gamma}, the optimal $\lambda$ values on AMusic dataset is 0.5 and the optimal $\gamma$ is 0.1.
Take NCE + ProRec for example (shown in Figure \ref{fig:lambda}-\ref{fig:gamma}), the optimal $\lambda$ and $\gamma$ on AMusic are found to be 0.5 and 0.1, respectively. In practice, we can initialize $\lambda$ to 0.5 for denser datasets and 0.25 for the sparser datasets with slight tuning. In general, we can set $\gamma$ to 0.1 .

\subsection{Denoising Case Study \textbf{(RQ4)}}
\begin{figure*}
\centering
\includegraphics[width=0.8\linewidth]{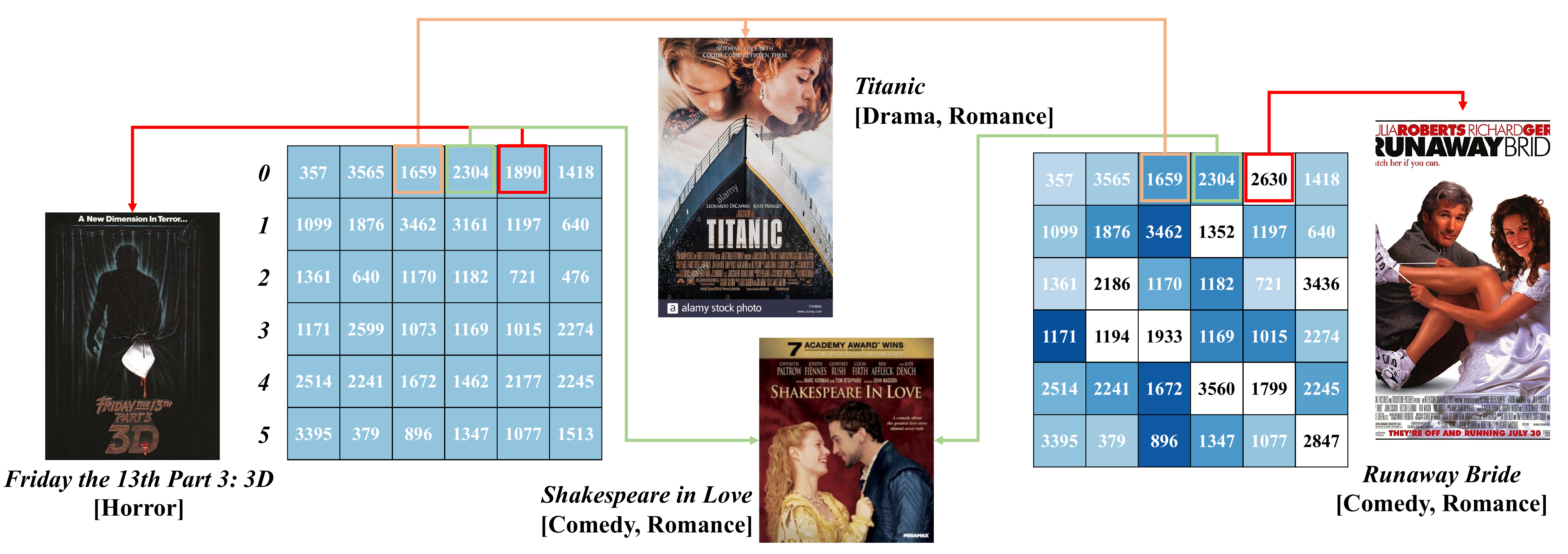}
\caption{The denoising process of ProRec.}
\label{fig:case}
\end{figure*}
To demonstrate the advantages of the denoising process,
we visualize the interaction matrix given by ML-100k on the left and the learned matrix of the proposed NCE + ProRec on the right (as shown in Figure \ref{fig:case}). 
The numbers in the boxes are actual movie IDs.
The colors in the original interaction matrix (left) are the same, indicating uniform weights, whereas those in the relabeled matrix (right) are different. The different depths of blue represent the matching cost (the deeper, the lower) and white boxes denote the intrinsically preferred items.
As we can see from the different colors, our many-to-many matching matrix effectively discriminates intrinsic and noisy interactions.
Moreover, we can also observe the personalized thresholding mechanism to work as different numbers of items being replaced for different users.
Furthermore, we show several movies relevant to user 0. Since most interacted movies of user 0 are from the \textsf{romance} category, such as \textit{Titanic} and \textit{Shakespeare in Love}, the \textsf{horror} movie is abnormal and is unlikely to reflect her/his intrinsic preference. NCE + ProRec automatically learns to highlight the two \textsf{romantic} movies.
After down-weighing \textit{Friday the 13th Part 3: 3D}, the scoring of user's preference changes, and thus the ranking of another \textsf{romantic} movie (\textit{Runaway Bride}) rise to show up in the top $K=6$ candidate list for user 0.

\section{Conclusion}
In this paper, we propose a novel ProRec framework for recommendation, which can effectively denoise the user-item interactions for both non-deep and deep learning RS methods with implicit feedback, so as to adaptively improve the recommendation accuracy.
We demonstrate the superior performance of ProRec in recommendation through extensive experiments and showcase its effectiveness in denoising user-item interactions through insightful case studies.
% Since we achieve denoising in a fully unsupervised fashion without accessing any additional user or item data, we expect ProRec to incur no more negative societal impact than any basic recommender system.

\bibliographystyle{named}
\bibliography{ProRec}

\vspace*{-8ex}
\begin{IEEEbiography}[{\includegraphics[width=0.96in,height=1.1in,clip,keepaspectratio]
{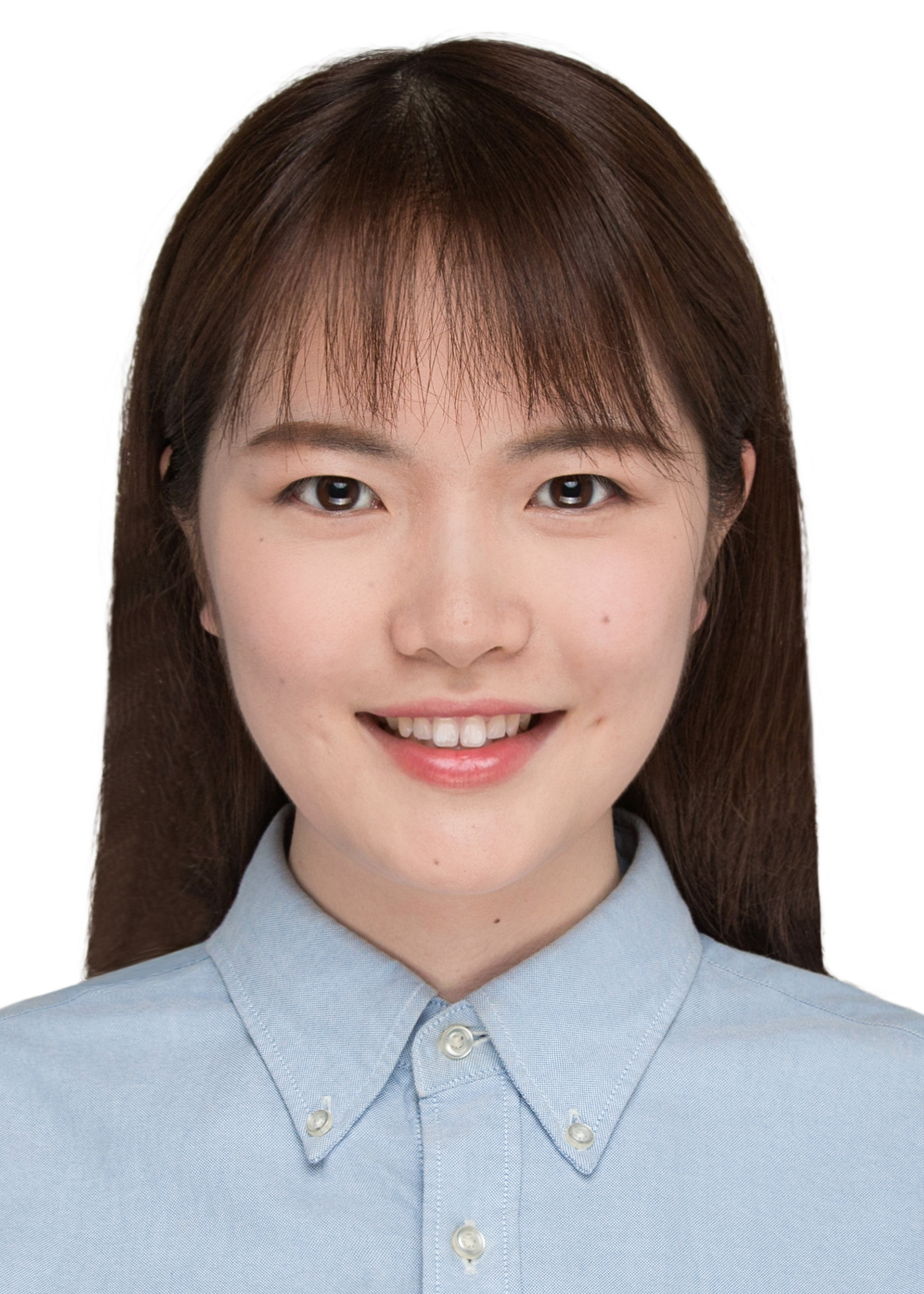}}]{Yanchao Tan}
received the BS degree in computer science from East China Normal University, China, in 2017. She is currently working toward the PhD degree in the College of Computer Science, Zhejiang University, China. Her research interests include recommender system and data mining.
\end{IEEEbiography}

\vspace*{-8ex}
\begin{IEEEbiography}[{\includegraphics[width=0.96in,height=1.1in,clip,keepaspectratio]
{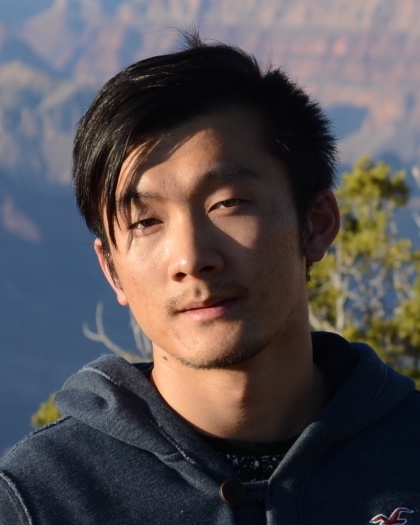}}]{Carl Yang} is an Assistant Professor in Emory University. He received his Ph.D.~in Computer Science at University of Illinois, Urbana-Champaign in 2020, and B.Eng.~in Computer Science and Engineering at Zhejiang University in 2014. His research interests span graph data mining, applied machine learning, structured information systems, with applications in knowledge graphs, recommender systems, biomedical informatics and healthcare. Carl's research results have been published in top venues like TKDE, KDD, WWW, NeurIPS, ICML, ICDE, SIGIR. He also received the Dissertation Completion Fellowship of UIUC in 2020 and the Best Paper Award of ICDM 2020.
\end{IEEEbiography}

\vspace*{-8ex}
\begin{IEEEbiography}[{\includegraphics[width=0.96in,height=1.1in,clip,keepaspectratio]
{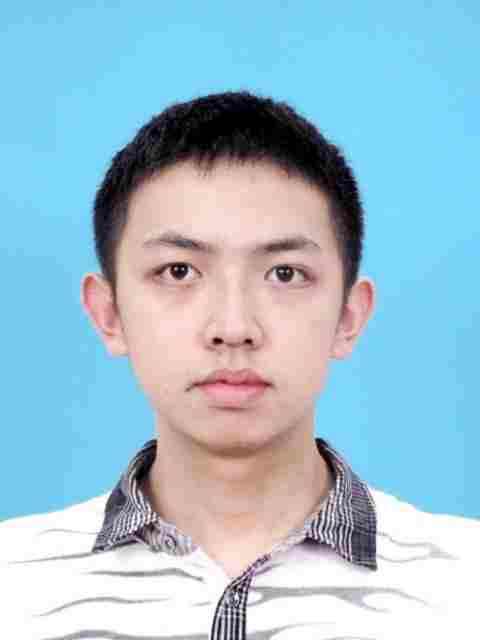}}]{Xiangyu Wei}
received the BS degree in software engineering from Zhejiang University, China, in 2019. He is currently working toward the Master degree in computer science, Zhejiang University, China. His research interests include recommender system and data mining.
\end{IEEEbiography}

\vspace*{-8ex}
\begin{IEEEbiography}[{\includegraphics[width=0.96in,height=1.1in,clip,keepaspectratio]
{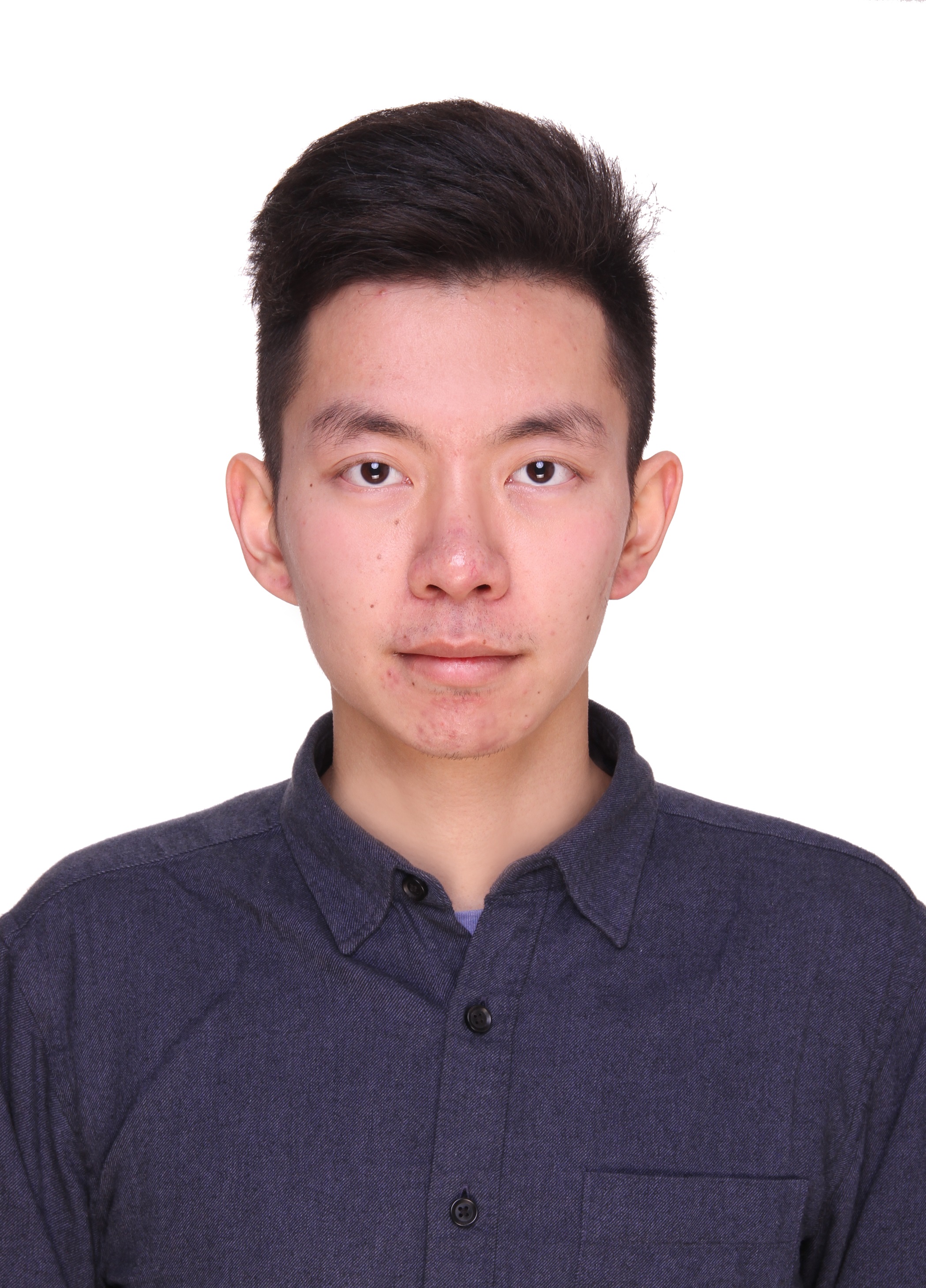}}]{Ziyue Wu}
received the BS degree in Statistics \& Information Management and Information Systems from Zhejiang University, China, in 2021. He is currently working toward the PhD degree in Management Science and Engineering, Zhejiang University, China. His research interests include econometrics and interpretable machine learning.
\end{IEEEbiography}

\vspace*{-8ex}
\begin{IEEEbiography}[{\includegraphics[width=0.96in,height=1.1in,clip,keepaspectratio]
{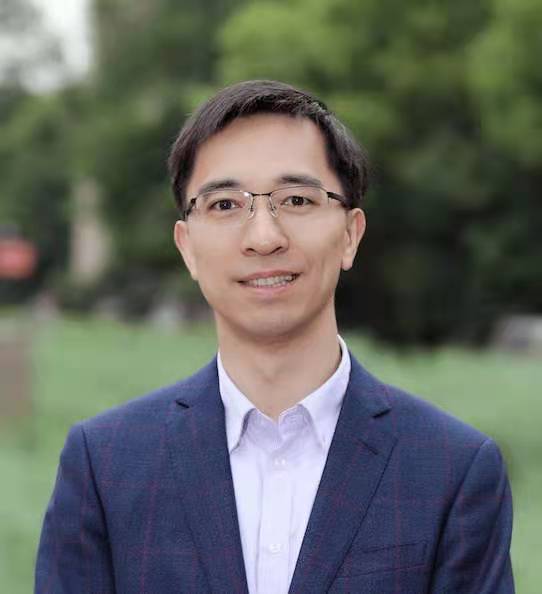}}]{Xiaolin Zheng}
PhD, Professor, PhD supervisor, and the deputy director of Institute of Artificial Intelligence, Zhejiang University. Senior member of IEEE and China Computer Federation, a committee member in Service
Computing of China Computer Federation, His main research interests include data mining, recommender system, and Intelligent finance.
\end{IEEEbiography}

\end{document}